\documentclass[a4paper,UKenglish,hyperref,cleveref, autoref, thm-restate]{lipics-v2021}
\title{Robust Positivity Problems for Linear Recurrence Sequences}
\titlerunning{Robust Positivity for LRS}

\subtitle{The frontiers of decidability for explicitly given neighbourhoods}

\author{Mihir Vahanwala}{Max Planck Institute for Software Systems, Saarland Informatics Campus, Germany}{mvahanwa@mpi-sws.org}{https://orcid.org/0009-0008-5709-899X}{}

\authorrunning{M.\ Vahanwala} 

\Copyright{M.\ Vahanwala} 



\nolinenumbers

\usepackage{amsfonts,stmaryrd,mathtools}
\usepackage{amsmath}
\usepackage{amssymb}
\usepackage{wrapfig}
\usepackage[ruled,vlined,linesnumbered,titlenotnumbered,noend]{algorithm2e}
\usepackage{xcolor}
\usepackage{booktabs}   
\usepackage{subcaption} 
\usepackage{pgf}
\usepackage{tikz}
\usepackage{listings}
\usepackage[colorinlistoftodos,prependcaption,textsize=tiny]{todonotes}
\usepackage{float}
\usepackage{tcolorbox}
\usepackage{bbm}
\usepackage{hyperref}
\usepackage{cleveref}
\usepackage{lineno}

\newcommand{\rationals}{\mathbb{Q}}
\newcommand{\naturals}{\mathbb{N}}
\newcommand{\integers}{\mathbb{Z}}
\newcommand{\reals}{\mathbb{R}}
\newcommand{\complexes}{\mathbb{C}}
\newcommand{\algebraics}{\overline{\mathbb{Q}}}
\newcommand{\realalgebraics}{\mathbb{A}}

\newcommand{\ball}{\mathcal{B}}

\newcommand{\seq}[1]{\langle#1\rangle}

\newtheorem{problem}{Problem}
\usepackage{graphicx}

\begin{document}
\maketitle
\begin{abstract}
Linear Recurrence Sequences (LRS) are a fundamental mathematical primitive for a plethora of applications such as the verification of probabilistic systems, model checking, computational biology, and economics. Positivity (are all terms of the given LRS non-negative?) and Ultimate Positivity (are all but finitely many terms of the given LRS non-negative?) are important open number-theoretic decision problems. Recently, the robust versions of these problems, that ask whether the LRS is (Ultimately) Positive despite small perturbations to its initialisation, have gained attention as a means to model the imprecision that arises in practical settings. However, the state of the art is ill-equipped to reason about imprecision when its extent is explicitly specified. In this paper, we consider Robust Positivity and Ultimate Positivity problems where the neighbourhood of the initialisation, expressed in a natural and general format, is also part of the input. We contribute by proving sharp decidability results: decision procedures at orders our techniques are unable to handle for general LRS would entail significant number-theoretic breakthroughs.
\end{abstract}
\section{Introduction}
\label{section:intro}
A real Linear Recurrence Sequence (LRS) of order $\kappa$ is an infinite sequence of real numbers $(u_0, u_1, u_2, \dots)$ having the following property: there exist $\kappa$ real constants $a_{0}, \dots, a_{\kappa-1}$, with $a_0 \ne 0$ such that for all $n \ge 0$:
\begin{equation}
u_{n+\kappa} = a_{\kappa-1}u_{n+\kappa-1} + \dots a_0 u_n.
\end{equation}
The constants $a_0, \dots, a_{\kappa-1}$ define the linear recurrence relation $\mathbf{a}$; they are also associated with the characteristic polynomial
$
X^\kappa - a_{\kappa-1}X^{\kappa-1} - \dots - a_1X - a_0.
$ 
The initial terms $u_0, \dots, u_{\kappa-1}$ are collectively denoted as the initialisation $\mathbf{c}$. An LRS is uniquely specified by $(\mathbf{a}, \mathbf{c})$. The best-known example is the Fibonacci sequence $\seq{0, 1, 1, 2, 3, 5, 8, \dots}$, satisfying the recurrence relation $u_{n+2} = u_{n+1} + u_n$: it is named after Leonardo of Pisa, who used it to model the population growth of rabbits. LRS have been extensively studied, and found several mathematical and scientific applications since. The monograph of Everest \textit{et al.} \cite{Everest2003RecurrenceS} is a comprehensive treatise on the mathematical aspects of Recurrence Sequences.

Important number-theoretic decision problems for Linear Recurrence Sequences include Positivity (is $u_n \ge 0$ for all $n$?), Ultimate Positivity (is $u_n \ge 0$ for all but finitely many $n$?) and the closely related Skolem Problem (is $u_n = 0$ for some $n$?). We remark that a Positive LRS is necessarily Ultimately Positive. These problems have applications in software verification, probabilistic model checking, discrete dynamic systems, theoretical biology, and economics. Decidability has been open for decades, with breakthroughs in restricted settings: Mignotte \textit{et al.\ }\cite{mignotte} and Vereshchagin \cite{vereshchagin} independently proved the Skolem Problem to be decidable up to order $4$. Ouaknine and Worrell \cite{joeljames3} showed Positivity and Ultimate Positivity are decidable up to order $5$ but number-theoretically hard at order $6$. For \textit{simple} LRS (those whose characteristic polynomials have no repeated roots), they showed that Positivity is decidable up to order $9$ \cite{ouaknine2014positivity} and Ultimate Positivity is decidable at all orders \cite{ouaknine2014ultimate}. These results were originally proven for LRS specified by \textit{rational} recurrences and initialisations, but can be generalised to real algebraic input as well. In this paper, we focus on Positivity and Ultimate Positivity for \textbf{sequences defined by real algebraic input}.

In contrast, the \textit{uninitialised} variants of these problems are far more tractable. Braverman \cite{Braverman06} and Tiwari \cite{Tiwari04} consider whether \textit{every} possible initialisation keeps the sequence Positive, and decide so in $\mathsf{PTIME}$. More recently, this result has been extended to processes with choices \cite{AGV18}. We argue that practical applications need a middle ground: recurrence relations that arise in practice need to be contextualised by actual instances of sequences; however, considering \textit{precise} initialisations does not account for inherently imprecise real world measurements, and the requirement of safety margins. We thus study robust variants: given a recurrence and an initialisation, do all initialisations in a neighbourhood satisfy (Ultimate) Positivity?

\paragraph*{Related Work} 
In this paper, we focus on the neighbourhood-of-initialisation notion of robustness, which was first introduced in \cite{originalstacs}, and more comprehensively treated in \cite{originalarxiv}.  Works with a more control-theoretic flavour include \cite{rounding20}, which allows for rounding at every step before applying the recurrence; in the same vein, \cite{pseudo21} allows for $\varepsilon$-disturbances at every step of the sequence. Our notion of robustness has been considered in \cite{originalstacs,originalarxiv,pseudo21}, however, these works primarily concern themselves with simply deciding whether there \textit{exists} a neighbourhood around the given point that satisfies Positivity, or whether there \textit{exists} a tolerance $\varepsilon$ such that the sequence avoids a region despite $\varepsilon$-disturbances at every step. Although they do identify that robust problems are hard when the neighbourhood is given as input, in the absence of decidability results, their hardness results are not sharp. 

There are, of course, broader approaches to model and reason about imprecision: \cite{N21} considers a model of computation that can take arbitrary real numbers as input, thereby allowing imprecision in both the initialisation and the recurrence. Even in this setting, the focus is on whether the decision is locally constant in \textit{some} neighbourhood of the given instance of the Positivity Problem, as opposed to whether the decision holds for an entire \textit{given} neighbourhood.

\paragraph*{Our contribution}
We address the gap in the robustness state of the art by exploring the frontiers of decidability when the neighbourhood is given as input. Concretely, our input consists of a linear recurrence relation $\mathbf{a}$ and a neighbourhood of initialisations centred around $\mathbf{c}$. Our problem is to decide whether all initialisations in the \textbf{given} neighbourhood result in (Ultimately) Positive sequences. 

When neighbourhoods are expressly given as input, their geometry plays a critical role in the decision procedure. The notion of neighbourhoods that we primarily focus on is based on the $\ell^2$-norm. We seek to slightly generalise the Euclidean $\varepsilon$-ball. More specifically, we use the Mahalanobis distance to define neighbourhoods. Our parameter is the positive definite matrix $\mathbf{S}$, and the neighbourhood of $\mathbf{c}$ it specifies is the set of all points $\mathbf{c'} \in \reals^{\kappa}$ such that $(\mathbf{c'} - \mathbf{c})^T\mathbf{S}(\mathbf{c'} - \mathbf{c}) \le 1$. The size of neighbourhoods is usually parametrised by an $\varepsilon$: in our case, we can account for it by simply scaling $\mathbf{S}$. In the statistical context, $\mathbf{S}$ is the inverse of a covariance matrix; and thus, our formulation is a rather natural way of capturing noise and measurement errors in the input, whose components may often be correlated. Our \textbf{novelty}, to the best of our knowledge, lies in identifying a general and practical way of explicitly specifying neighbourhoods, and establishing the \textbf{first decidability results} in such a setting, albeit at low orders or subject to spectral constraints.

As first discussed in \cite[Section 5]{joeljames3}, solving decision problems on Linear Recurrence Sequences in full generality is an endeavour fraught with number-theoretic hardness. Decision procedures for Positivity problems for LRS of higher order would allow number theorists to compute properties of irrational numbers that are considered inaccessible to contemporary techniques. These include the Diophantine approximation type, which intuitively describes the quality of the ``best'' rational approximation of a given irrational number, and the Lagrange constant, which intuitively describes how well increasingly precise rational approximations of a given irrational number converge. We justify the inability of our techniques to handle LRS of higher orders by \textbf{reducing the computation of Diophantine approximation types and Lagrange constants} to robust Positivity problems for LRS of lower orders than ever before.

\begin{table}[H]
\begin{tabular}{|l|r|r|r|}
  \hline
  & \multicolumn{2}{c|}{\bf Decidability Proof} &  \\
  \hline
   \textbf{Problem:} $\mathbf{S}$-\textbf{Robust}& \textbf{General}& {\bf Simple} &{\bf Hardness} \\
  \hline
  Positivity & order $\le 4$ & order $\le 5$ &Diophantine hard at order 5\\
  Uniform Ultimate Positivity & order $\le 4$ & \textbf{all orders} & Lagrange hard at order 5 \\
  Non-uniform Ultimate Positivity & order $\le 4$ & order $\le 4$ & \cite{joeljames3,originalarxiv}: Lagrange hard at order 6 
  \\
  \hline
\end{tabular}
\caption{Main results, summarised. The distinction between uniform and non-uniform refers to whether the threshold index for certifying Ultimate Positivity must be common for the entire neighbourhood.}%
\end{table}

\paragraph*{Structure of the paper}
The exponential polynomial closed form is an invaluable tool in the study of LRS, and we devote \S\ref{section:solspace} to its exposition. This equips us to introduce our Robust Positivity Problems and intuit their decidability proofs in \S\ref{section:problems}. Linear Recurrences and Diophantine Approximation are intrinsically connected: number-theoretic results form the basis of decision procedures; open problems are a yardstick for hardness reductions. We survey the number theory relevant to us in \S\ref{section:diophantine}. We then prove our decidability results in the technical \S\ref{section:decidability} and \S\ref{section:decidability2}, and present our hardness reduction in \S\ref{section:hardness}. We provide concluding perspective in \S\ref{section:perspective}. We refer the reader to Appendix \ref{appendix:prelims} for a summary of the standard notation and prerequisites we use.

\section{The exponential polynomial closed form}
\label{section:solspace}
We begin by discussing the exponential polynomial closed form, a perspective that is routinely leveraged to study the behaviour of Linear Recurrence Sequences. Simple LRS (no repeated characteristic roots) have the closed form
\begin{equation}
\label{eq:exppoly}
u_n = \sum_j w_j \rho_j^n + \sum_j (z_j \gamma_j^n + \bar{z_j}\bar{\gamma_j}^n)
\end{equation}
where each $\rho_j, \gamma_j, \bar{\gamma_j}$ are distinct roots of the characteristic polynomial. By straightforward arithmetic on the above expression, we can see that if $(u_n)_{n\in \naturals}, (v_n)_{n\in \naturals}$ are simple LRS with sets of characteristic roots $U$ and $V$ respectively, then
\begin{itemize}
\item $r_n = u_n + v_n$ is a simple LRS, whose set of roots is $U \cup V$.
\item $r_n = u_n \cdot v_n$ is a simple LRS, whose set of roots is $\{\gamma_1\gamma_2: \gamma_1 \in U, \gamma_2 \in V\}$.
\end{itemize}

In general, one can encode a linear recurrence $\mathbf{a}$ as a $\kappa \times \kappa$ companion matrix $\mathbf{A}$, and interpret the initialisation $\mathbf{c}$ as a vector. Then, $u_n$ is given by the first coordinate of $\mathbf{A}^n\mathbf{c}$, i.e.
\begin{equation}
\label{eq:companion}
\begin{bmatrix}
u_n \\
u_{n+1} \\
\vdots \\
u_{n+\kappa-1}
\end{bmatrix} 
= 
\begin{bmatrix}
0 & 1 & 0 & \dots & 0 \\
0 & 0 & 1 & \dots & 0 \\
\vdots & \vdots & \vdots & \dots & \vdots \\
a_0 & a_1 & a_2 & \dots & a_{\kappa-1}
\end{bmatrix}^n
\begin{bmatrix}
u_0 \\
u_{1} \\
\vdots \\
u_{\kappa-1}
\end{bmatrix}.
\end{equation}
Let $\mathbf{e_1}^T$ denote the row vector $\begin{bmatrix}1 & 0 & \dots & 0\end{bmatrix}$. We can thus write $u_n = \mathbf{e_1}^T\mathbf{A}^n\mathbf{c}$. It is now a standard fact that LRS have the following \textbf{real exponential polynomial} closed form 
\begin{equation}
\label{eq:realexppoly}
u_n = \left(\sum_{j=1}^{k_1}\sum_{\ell = 0}^{m_j-1} z_{j\ell}\rho_j^n n^\ell\right) + \left(\sum_{j=k_1 + 1}^{k_2} \sum_{\ell = 0}^{m_j-1} (x_{j\ell} \cos n\theta_j + y_{j\ell}\sin n\theta_j)\rho_j^n n^\ell\right)
\end{equation}
where $\rho_j$ (alternately, $\rho_j e^{i\theta_j}$) are roots of the characteristic polynomial defined by $\mathbf{a}$, each with multiplicity $m_j$. The coefficients $z_{j\ell}, x_{j\ell}, y_{j\ell}$ each depend linearly on $\mathbf{c}$. $u_n$ can thus be equivalently expressed as the inner (dot) product $\seq{\mathbf{p}, \mathbf{q_n}}$ where $\{\mathbf{q_n}\}_{n\in \naturals}$ is the sequence of vectors of terms that occur in the exponential polynomial expression, and $\mathbf{p}$ is the vector of corresponding coefficients. The choice of $\{\mathbf{q_n}\}_{n\in \naturals}$ can differ in ``phase'': one can replace $\cos n\theta, \sin n\theta$ by $\cos (n\theta-\varphi), \sin(n\theta-\varphi)$ for some choice of $\varphi$, and adjust the corresponding coefficients in $\mathbf{p}$ accordingly.

Roots such that $|\rho_j|$ is the largest are called \textbf{dominant}. The growth rate of a term in the above expression is governed by $\rho_j^n n^\ell$. Terms with the fastest growth are called \textbf{dominant terms}, and they drive the asymptotic behaviour of the LRS. A standard, intuitive prerequisite for Ultimate Positivity is that the leading terms in the exponential polynomial expression must include one that is real and strictly positive, otherwise their dominant contribution oscillates between positive and negative. It is formalised by applying \cite[Lemma 4]{Braverman06} to the dominant terms in expression \ref{eq:realexppoly} and arguing that the contribution from the remaining terms vanishes asymptotically. 
\begin{proposition}
\label{prop:folklore}
If the characteristic polynomial has no real dominant root of maximum multiplicity, then in any full-dimensional neighbourhood of initialisations, there exists an initialisation, such that the sequence has infinitely many positive terms, and infinitely many negative terms.
\end{proposition}

\textbf{Henceforth, we assume that the characteristic polynomial has a real positive dominant root of maximum multiplicity, for otherwise the answer to Ultimate Positivity is trivially NO.} 

We define $\seq{\mathbf{p}, \mathbf{q_n}}_{dom}$ to be the normalised contribution of the dominant terms in the exponential polynomial solution. That is, if the dominant growth rate is $\rho^n n^\ell$, we pick terms with that growth rate, and divide their contribution by $\rho^n n^\ell$. For example, if
$
u_n = p_1 2^n + p_2 2^n\cos (n\theta -\varphi) + p_3 2^n \sin (n\theta-\varphi) + p_4
$
then $$\seq{\mathbf{p}, \mathbf{q_n}}_{dom} = p_1  + p_2 \cos (n\theta-\varphi) + p_3 \sin (n\theta-\varphi).$$
We define
$
\mu(\mathbf{c}) = \liminf_{n \in \naturals} \seq{\mathbf{p}, \mathbf{q_n}}_{dom}.
$
Note that $\mu$ is an intrinsic property of the initialisation $\mathbf{c}$ and the sequence it generates, and hence is invariant under the choice of ``phase shift'' $\varphi$ while defining $\{\mathbf{q_n}\}_{n\in \naturals}$. In our example, it is $p_1 - \sqrt{p_2^2 + p_3^2}$.

\section{Robust Positivity Problems}
\label{section:problems}

In this paper, we shall focus on defining and tackling Robust Positivity problems. Our input consists of a linear recurrence relation $\mathbf{a}$, an initialisation $\mathbf{c}$, and a positive definite matrix $\mathbf{S}$ that is used to define a neighbourhood around $\mathbf{c}$. \textbf{All input is real algebraic.}\footnotemark \footnotetext{The field of algebraic numbers $\algebraics$ is the algebraic closure of the rationals $\rationals$. Arithmetic and polynomial factoring over $\algebraics$ can be performed with exact precision. We use $\realalgebraics$ to denote the field of real algebraic numbers, and refer the reader to Appendix \ref{appendix:prelims} for an initiation to these number fields.}

\begin{problem}[$\mathbf{S}$-Robust Positivity]
\label{prob:rrobpos}
Decide whether for all $\mathbf{c'}$ such that $(\mathbf{c'} - \mathbf{c})^T\mathbf{S}(\mathbf{c'} - \mathbf{c}) \le 1$, the LRS $(\mathbf{a}, \mathbf{c'})$ is positive.
\end{problem}

\begin{problem}[$\mathbf{S}$-Robust Uniform Ultimate Positivity]
\label{prob:rrobuniultpos}
Decide whether there exists an $N$ such that for all $\mathbf{c'}$ with $(\mathbf{c'} - \mathbf{c})^T\mathbf{S}(\mathbf{c'} - \mathbf{c}) \le 1$, the LRS $(\mathbf{a}, \mathbf{c'})$ is positive from the $N^{th}$ term onwards.
\end{problem}

We can switch the order in which $N$ and $\mathbf{c'}$ are quantified, and query a weaker notion of Robust Ultimate Positivity:
\begin{problem}[$\mathbf{S}$-Robust Non-uniform Ultimate Positivity]
\label{prob:rrobnonuniultpos}
Decide whether for all $\mathbf{c'}$ with $(\mathbf{c'} - \mathbf{c})^T\mathbf{S}(\mathbf{c'} - \mathbf{c}) \le 1$ , there exists an $N$ such that the LRS $(\mathbf{a}, \mathbf{c'})$ is positive from the $N^{th}$ term onwards.
\end{problem}

The attentive reader might have already noticed that we depart from convention and specify neighbourhoods as \textit{closed} balls. Although \cite{originalarxiv} does not solve the problems we consider in this paper, it makes crucial observations about the geometry: for Problems \ref{prob:rrobpos} and \ref{prob:rrobuniultpos}, there is no difference between open and closed balls. On the other hand, Problem \ref{prob:rrobnonuniultpos} becomes considerably easier with open balls, and its decidability in this case is tackled in \cite{originalarxiv} itself. 

\subsection{Uniform Variants: The foundation}
\label{section:uniformfoundation}
In general, an arbitrary point $\mathbf{c'}$ is expressed as $\mathbf{c} + \mathbf{d}$, where $\mathbf{d} \in \mathcal{P}$, a full-dimensional neighbourhood symmetric about the origin. Observe equation \ref{eq:companion}. The $n^{th}$ term of the LRS is non-negative throughout the neighbourhood if and only if for all $d \in \mathcal{P}$,
$
\mathbf{e_1}^T \mathbf{A}^n (\mathbf{c + d}) \ge 0.
$
We can use the symmetry of $\mathcal{P}$ about the origin to rewrite the above as
\begin{equation}
\label{eq:illustrate}
\mathbf{e_1}^T \mathbf{A}^n \mathbf{c}\ge \max_{\mathbf{d} \in \mathcal{P}} \mathbf{e_1}^T\mathbf{A}^n\mathbf{d} \ge 0.
\end{equation}
As a simple illustration, assume that the neighbourhood is defined by a polytope rather than a positive definite matrix. This situation arises, for instance, when the metric is based on the $\ell^1$- or $\ell^\infty$-norm, as opposed to the $\ell^2$-norm.  In this simple example, $\mathcal{P}$ is a polytope, hence $\mathbf{e_1}^T\mathbf{A}^n\mathbf{d}$ is maximised at one of the finitely many corners $\{\mathbf{d_1}, \dots, \mathbf{d_k}\}$. Thus, Robust (Uniform Ultimate) Positivity is decided by using the state of the art \cite{joeljames3} to check the (Ultimate) Positivity of each of the LRS $(\mathbf{a}, \mathbf{c+d_i})$. The geometry of our setting is not simple enough to allow such a straightforward approach.
\textbf{The overview of our approach to Problems \ref{prob:rrobpos} and \ref{prob:rrobuniultpos} is as follows.}
\begin{enumerate}
\item Decide (constructively for Problem \ref{prob:rrobpos}) whether there exists an $N_1$ such that $\mathbf{e_1}^T \mathbf{A}^n \mathbf{c} \ge 0$ for all $n > N_1$. If $N_1$ is explicitly required, the state of the art is able to tackle LRS of order $\le 5$ \cite{joeljames3} and Simple LRS of order $\le 9$ \cite{ouaknine2014positivity}. In the non-constructive case, it can further handle all simple LRS \cite{ouaknine2014ultimate}.
\item Use linear-algebraic arguments to define a real algebraic LRS $(v_n)_{n=0}^\infty$, such that $v_n \ge 0$ if and only if $|\mathbf{e_1}^T \mathbf{A}^n \mathbf{c}|\ge \max_{\mathbf{d} \in \mathcal{P}} \mathbf{e_1}^T\mathbf{A}^n\mathbf{d}$.
\item Decide (constructively for Problem \ref{prob:rrobpos}) whether there exists $N_2$ such that $v_n \ge 0$ for all $n > N_2$. Positivity throughout the neighbourhood is thus guaranteed beyond step $N = \max(N_1, N_2)$. If either $N_1$ or $N_2$ does not exist, then Robust Ultimate Positivity, and hence Robust Positivity, does not hold.
\item \textbf{Only for Problem \ref{prob:rrobpos}:} Explicitly check inequality \ref{eq:illustrate} for $n \le N$.
\end{enumerate}

Our novelty lies in Step 2 and identifying when Step 3 can be implemented. We now discuss how we perform Steps 2 and 4 when $\mathcal{P} = \mathcal{B}_\mathbf{S}$, a neighbourhood of vectors $\mathbf{d}$ such that $\mathbf{d}^T\mathbf{S}\mathbf{d} \le 1$. The defining parameter $\mathbf{S}$ is a real algebraic positive definite matrix. We note that since $\mathbf{S}$ is positive definite, it can be factored as $\mathbf{G}^T\mathbf{G}$, where $\mathbf{G}$ is a real algebraic invertible matrix. We denote $\mathbf{Gd} = \mathbf{f}$. We argue that $\mathbf{G}^{-1}$ bijectively maps the Euclidean unit ball $\mathcal{B}$ to $\mathcal{B}_\mathbf{S}$. The bijection is clear from the invertibility of the matrix. Suppose $\mathbf{d} = \mathbf{G}^{-1}\mathbf{f}$, where $\mathbf{f} \in \mathcal{B}$, i.e. $\mathbf{f}^T\mathbf{f} \le 1$. Then $\mathbf{d}^T\mathbf{Sd} = \mathbf{d}^T\mathbf{G}^T\mathbf{Gd} = \mathbf{f}^T\mathbf{f} \le 1.$
Hence,
\begin{equation}
\label{eq:bijectivemap}
\max_{\mathbf{d} \in \mathcal{B}_\mathbf{S}} \mathbf{e_1}^T\mathbf{A}^n\mathbf{d} = \max_{\mathbf{f} \in \mathcal{B}} \mathbf{e_1}^T\mathbf{A}^n\mathbf{G}^{-1}\mathbf{f}.
\end{equation}
$\mathcal{B}$ is a convex set; thus a linear function will necessarily be maximised at its boundary, i.e. when $||\mathbf{f}|| = 1$. The linear function $\mathbf{h}^T\mathbf{f}$ is maximised over the unit Euclidean ball when $\mathbf{f}$ is aligned along $\mathbf{h}$; the maximum value is $||\mathbf{h}||$. We can thus perform Step 4 because
\begin{equation}
\max_{\mathbf{d} \in \mathcal{B}_\mathbf{S}} \mathbf{e_1}^T\mathbf{A}^n\mathbf{d} = \left|\left|\left( \mathbf{e_1}^T\mathbf{A}^n\mathbf{G}^{-1} \right)^T\right|\right|.
\end{equation}

For Step 2, we need a necessary and sufficient condition for $|\mathbf{e_1}^T \mathbf{A}^n \mathbf{c}|\ge \max_{\mathbf{d} \in \mathcal{P}} \mathbf{e_1}^T\mathbf{A}^n\mathbf{d}$, in terms of the positivity of an LRS at iterate $n$. We simply square both sides of the inequality, and transfer all terms to the left: 
\begin{equation}
\label{eq:critical}
(\mathbf{e_1}^T \mathbf{A}^n \mathbf{c})^2 - (\mathbf{e_1}^T \mathbf{A}^n \mathbf{g_1})^2 - \dots - (\mathbf{e_1}^T \mathbf{A}^n \mathbf{g_\kappa})^2 \ge 0.
\end{equation}
Crucially, $\mathbf{g_1}, \dots, \mathbf{g_\kappa}$ are the linearly independent columns of the invertible $\mathbf{G}^{-1}$. Only Step 3 remains to be addressed: we must (constructively) decide whether there exists $N_2$ such that the previous inequality holds for all $n > N_2$. In \S\ref{section:decidability}, we give the technical details, thus proving our first main decidability result.

\begin{theorem}[First Main Decidability Result]
\label{thm:decide}
Problem \ref{prob:rrobuniultpos} is decidable for simple LRS. Problem \ref{prob:rrobpos} is decidable for simple LRS up to order 5. Problems \ref{prob:rrobpos} and \ref{prob:rrobuniultpos} are decidable for general LRS up to order 4.
\end{theorem}

\subsection{The non-uniform variant: An overview}
\label{section:nonuniformoverview}
As discussed at length in \cite{originalstacs,originalarxiv}, $\mu(\mathbf{c'})= \liminf_{n \in \naturals} \seq{\mathbf{p'}, \mathbf{q_n}}_{dom} \ge 0$ is necessary for the Ultimate Positivity of $\mathbf{c'}$; $\mu(\mathbf{c'}) > 0$ is sufficient.

\textbf{Our strategy for Problem \ref{prob:rrobnonuniultpos} is as follows.}
\begin{enumerate}
\item Use the First Order Theory of the Reals to check that $\mu(\mathbf{c'}) \ge 0$ for all $\mathbf{c'}$ in the given neighbourhood, and detect the critical boundary cases when $\mu(\mathbf{c'}) = 0$.
\item Exploit the low dimensionality to decide the critical boundary cases when $\mu(\mathbf{c'}) = 0$.
\end{enumerate}

\begin{figure}[h]

\includegraphics[width=\textwidth]{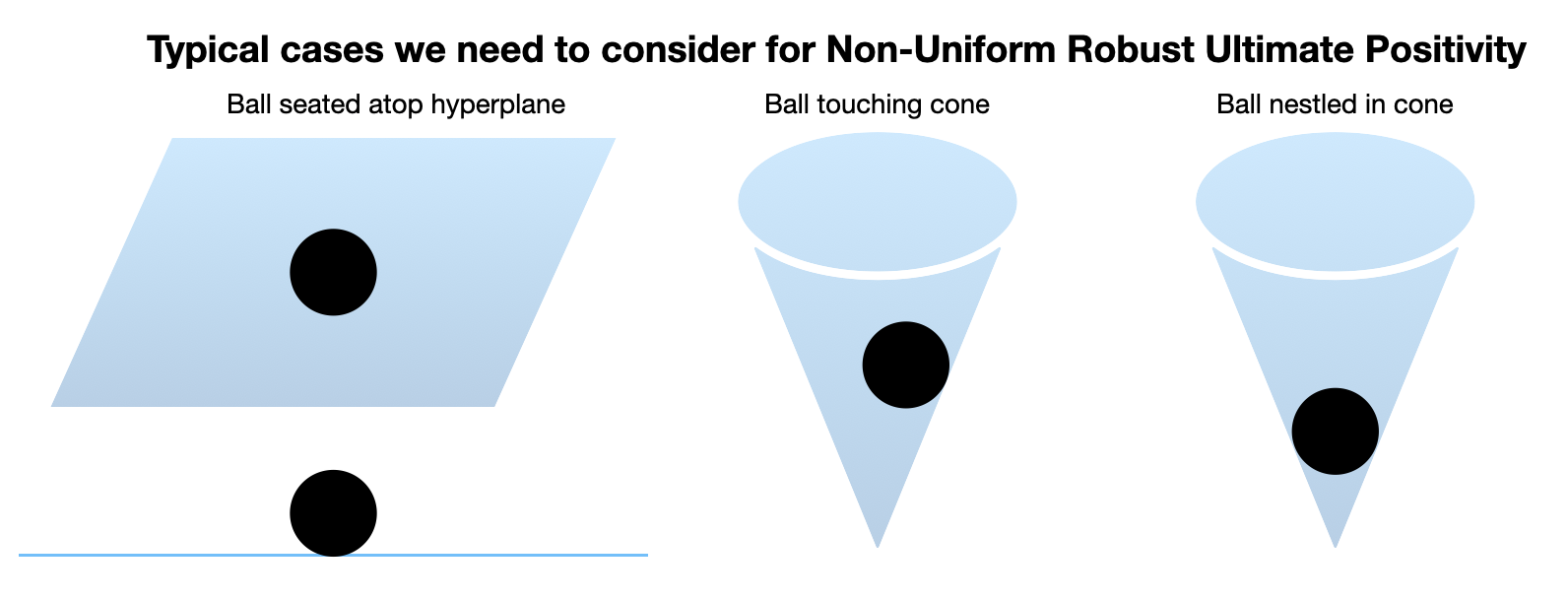}
\caption{Visual intuition. The region $\mu \ge 0$ is defined by the intersection of halfspaces. The orientation of the neighbourhood relative to this region is deduced with the First Order Theory of the Reals. When there are finitely many halfspaces, the critical case is marked by the ball being tangent to the separating hyperplane(s) at finitely many discrete points. In low dimensions, Ultimate Positivity can be decided for these boundary cases using existing techniques. When there are infinitely many halfspaces, they carve out a region that resembles a cone. The neighbourhood can either touch the cone as before, or be nestled in it, having a continuous, connected region of tangency. In the latter case, Robust Ultimate Positivity can be handled with number-theoretic arguments in the low-dimensional setting.}
\label{fig:geometricpicture}
\end{figure}

We adopt this strategy (see Figure \ref{fig:geometricpicture}) and prove our second decidability result in \S\ref{section:decidability2}.
\begin{theorem}[Second Decidability Result]
\label{thm:decide2}
Problem \ref{prob:rrobnonuniultpos} is decidable up to order 4.
\end{theorem}

\section{Diophantine Approximation}
\label{section:diophantine}
We justify the inability of our techniques to generalise to LRS of higher order by establishing a connection to a number-theoretic hurdle: that of Diophantine Approximation. Diophantine Approximation is a vast and active number-theoretic field of research, one of whose concerns is the approximation of reals by rational numbers. A key tool in this regard is the partial fraction expansion $[a_0; a_1, a_2, \dots]$ of an irrational $t$:
$$
t = a_0 + \cfrac{1}{a_1 + \cfrac{1}{a_2 + \cfrac{1}{\dots}}}
$$
where $a_0, a_1, a_2, \dots \in \naturals$. Truncating this expansion at progressively greater depths yields a series of increasingly accurate approximations. The quality of the rational approximation depends not only on its accuracy but also on the size of the denominator. As discussed in the Introduction, evaluating the quality of the approximation, or that of the convergence, seems inaccessible to contemporary number theory.

The above intuition about the quality of the approximation is captured in the following definition of $L(t)$, the (homogenous) Diophantine approximation type:
\begin{equation}
L(t) = \inf\left\{c \in \reals : \left|t - \frac{p}{q}\right| < \frac{c}{q^2} \text{ for some } p, q\in \integers\right\}.
\end{equation}
Similarly, the quality of the convergence is formalised by defining $L_{\infty}(t)$, the (homogenous) Lagrange constant: 
\begin{equation}
L_\infty(t) = \inf\left\{c \in \reals : \left|t - \frac{p}{q}\right| < \frac{c}{q^2} \text{ for infinitely many } p, q\in \integers\right\}.
\end{equation}
 
For technical purposes, we use an equivalent definition that relates to the continued fraction perspective, and allows for a slight generalisation. We follow Lagarias and Shallit’s terminology \cite{dio-constants} and use $[x]$ to denote the shortest distance from $x$ to an integer; while $[x]_b$ denotes the shortest distance from $x$ to an integer multiple of $b \in \reals$. It is easy to observe the property $[x]_b = b[x/b]$.

\begin{definition}[Diophantine Approximation Type]
\label{def:L}
The homogenous Diophantine approximation type $L(t)$ is defined to be $\inf_{n \in \naturals_{>0}} n[nt]$. The inhomogeneous Diophantine approximation type $L(t, s)$ is defined to be $\inf_{n \in \naturals_{>0}} n[nt - s]$, $s \notin \integers + t\integers$. 
\end{definition} 

\begin{definition}[Lagrange constant]
\label{def:Linfty}
The homogenous Lagrange constant $L_\infty(t)$ is defined to be $\liminf_{n \in \naturals} n[nt]$. The inhomogeneous Lagrange constant $L_\infty(t, s)$ is defined to be\\ $\liminf_{n \in \naturals} n[nt - s]$, $s \notin \integers + t\integers$.
\end{definition} 

From the definitions, it is clear that $0 \le L(t) \le L_\infty(t)$. Due to the work of Khintchine \cite{khintchine}, it is known that these constants lie between $0$ and ${1}/{\sqrt{5}}$. In our setting, the irrational $t$ comes from the argument $\theta$ of the characteristic root $\rho e^{i\theta}$ of the LRS: $t = \theta/2\pi$. The following observation hints at how the correspondence could further extend to Positivity problems. It is pivotal to our novel low-dimensional decidability result for Uniform Robustness.
\begin{lemma}
\label{eq:quadraticdecay}
For all $\theta$ that are not rational multiples of $2\pi$, and all $\varphi$, there exist infinitely many $n\in \naturals$ such that
$
1 - \cos(n\theta- \varphi) \le \frac{1}{2}\left[n\theta - \varphi \right]_{2\pi}^2 = 2\pi^2[nt -s]^2 \le \frac{2\pi^2}{5n^2}.
$
\end{lemma}

The properties of the LRS are driven by whether the characteristic root $\rho e^{i\theta}$ is a root of unity, i.e. $\theta$ is a rational multiple of $2\pi$. If yes, decision procedures are often much simpler; if not, we appeal to the number theory discussed in this section. One can detect whether an algebraic characteristic root is a root of unity by brute enumeration.

\begin{lemma}
\label{lemma:rootofunity}
Let $\alpha$ be an algebraic number of degree $d$. Then if $\alpha$ is a $k^{th}$ root of unity, $k \le 2d^2$.
\end{lemma}
\begin{proof}
The degree of the $k^{th}$ root of unity is precisely $\Phi(k)$, where $\Phi$ denotes the Euler totient function. $\Phi(k) \ge \sqrt{k/2}$. The desired inequality follows.
\end{proof}

We record a number-theoretic fact which describes the density of the integer multiples of an irrational $x$ modulo $1$ in the unit interval: its proof relies on continued fraction expansions and the Ostrowski numeration system \cite{bourla2016ostrowski,berthe2022dynamics}, and is deferred to Appendix \ref{appendix:ostrowski}. This result is decisive when considering Non-Uniform Robustness.
\begin{lemma}
\label{lemma:existsreal}
For every irrational number $x$, strictly decreasing real positive function $\psi$, and interval $\mathcal{I} = [a, b] \subset [0, 1], ~ a \ne b$, there exists $y_0 \in \mathcal{I}$ such that $[nx - y_0] < \psi(n)$ for infinitely many even $n$, and $y_1 \in \mathcal{I}$ such that $[nx - y_1] < \psi(n)$ for infinitely many odd $n$.
\end{lemma}

The familiar density theorem is an immediate corollary of the above powerful result. Indeed, we can consider an interval of length $\varepsilon/2$, and take $\psi(n) = \varepsilon/2$.
\begin{lemma}
\label{lemma:density}
Let $x$ be irrational, and $y \in [0, 1)$. For every $\varepsilon > 0$, there exist infinitely many even $n$, and infinitely many odd $n$ such that $[nx - y] < \varepsilon$.
\end{lemma}

The following application of the density theorem is central to the computation of the discrete $\mu(\mathbf{c}) = \liminf_{n\in \naturals}\seq{\mathbf{p}, \mathbf{q_n}}_{dom}$.
\begin{lemma}
\label{eq:liminfmin}
Suppose $\theta$ is not a rational multiple of $2\pi$. Let $h_1(t)$ be continuous with period $2\pi$. Then
$
\liminf_n \left(h_1(n\theta) + h_2(-1^n)\right) = \min_{t \in [0, 2\pi], b\in \{-1, 1\}} \left(h_1(t) + h_2(b)\right).
$
\end{lemma}

We note that despite the observations and results mentioned in the preceding discussion, the Diophantine approximation type and Lagrange constant of most transcendental numbers are unknown. For instance, computing $L_\infty(\pi)$ is a longstanding and mathematically interesting open problem. We refer the reader to \cite[Section 5]{joeljames3} for a cursory survey of the history of relevant developments in the field of Diophantine approximation. This source reduces the computation of the constants $L(t)$ and $L_\infty(t)$ discussed above to the non-robust variants of Positivity problems for LRS of order 6. In \S\ref{section:hardness}, we prove analogous hardness results for robust Positivity problems of order 5. To that end, we define a similar class of transcendental numbers relevant to our reduction. Let
\begin{equation}
\mathcal A=\{p+q i \in \mathbb{C} \mid p,q \in \mathbb{A}, p^2+q^2=1, \forall n.~(p + qi)^n \ne 1\}
\end{equation}
i.e., the set $\mathcal A$ consists of algebraic numbers on the unit circle in $\complexes$, none of which are roots of unity. In particular, writing $p+q i= e^{i 2 \pi \theta}$, we have that $\theta \notin \mathbb{Q}$. We denote:
\begin{equation}
\label{eq:keyset}
\mathcal{T} = \left\{ \theta \in (- 1/2, 1/2] \mid e^{2 \pi i \theta} \in \mathcal{A}\right\}.
\end{equation}
The set $\mathcal{T}$ is dense in $(- \frac 1 2, \frac 1 2]$. In general, we don't have a method to compute $L(\theta)$ or $L_\infty(\theta)$ for $\theta \in \mathcal{T}$, or approximate them to arbitrary precision.

\begin{definition}[Number-theoretic hardness]
\label{def:hardness}
Let $\mathcal{T}$ be as above. A decision problem is said to be $\mathcal{T}$-Diophantine hard (resp.\ $\mathcal{T}$-Lagrange hard), if its decidability entails that given any $t \in \mathcal{T}$ and $\varepsilon > 0$, one can compute $\ell$ such that $|\ell - L(t)| < \varepsilon$ (resp.\  $|\ell - L_\infty(t)| < \varepsilon$).
\end{definition}

\begin{theorem}[Main Hardness Result]
\label{thm:hardness}
Problem \ref{prob:rrobpos} (resp.\ Problem \ref{prob:rrobuniultpos}) is $\mathcal{T}$-Diophantine hard (resp.\ $\mathcal{T}$-Lagrange hard) at order 5. 
\end{theorem}

As noted in \cite{originalarxiv}, in view of the Lagrange hardness (Definition \ref{def:hardness}) of Ultimate Positivity at order 6 \cite{joeljames3}, Problem \ref{prob:rrobnonuniultpos}, which asks whether the given neighbourhood consists entirely of initialisations that produce an Ultimately Positive sequence, is also Lagrange hard at order 6. The idea is to use the existing reduction from the computation of Lagrange constants to Ultimate Positivity, and extend it to the robust variant: one simply constructs a neighbourhood of initialisations that has the hard instance of Ultimate Positivity on its surface, but otherwise lies entirely in the region where Ultimate Positivity is guaranteed.

\begin{theorem}
\label{thm:hardness2}
Problem \ref{prob:rrobnonuniultpos} is $\mathcal{T}$-Lagrange hard at order 6. 
\end{theorem}

\section{Decidability of Uniform Robustness}
\label{section:decidability}

  \label{thm:abelian}

In this section, we prove Theorem \ref{thm:decide} by showing that we can implement Step 3 of the overview in \S\ref{section:uniformfoundation}: (constructively) decide whether there exists $N_2$ such that for all $n > N_2$,
\begin{equation}
\label{eq:criticalcopy}
(\mathbf{e_1}^T \mathbf{A}^n \mathbf{c})^2 - (\mathbf{e_1}^T \mathbf{A}^n \mathbf{g_1})^2 - \dots - (\mathbf{e_1}^T \mathbf{A}^n \mathbf{g_4})^2 \ge 0.
\end{equation}

\begin{theorem}[First Main Decidability Result, restated]
\label{thm:decidecopy}
Problem \ref{prob:rrobuniultpos} (Robust Uniform Ultimate Positivity) is decidable for simple LRS. Problem \ref{prob:rrobpos} (Robust Positivity) is decidable for simple LRS up to order 5. Problems \ref{prob:rrobpos} and \ref{prob:rrobuniultpos} are decidable for general LRS up to order 4.
\end{theorem}

\subsection{Simple LRS}
We begin by treating simple LRS. The goal is to show that the current state of the art is equipped to handle instances relevant to this setting. Recall the discussion on the point-wise sums of products of simple LRS, surrounding equation \ref{eq:exppoly}. If the original LRS is simple, then inequality \ref{eq:criticalcopy} is also an instance of Ultimate Positivity for simple LRS; indeed, its input can be seen to be real algebraic. In case we are only interested in Robust Ultimate Positivity, the non-constructive decision procedure \cite{ouaknine2014ultimate} suffices, because it completely solves Ultimate Positivity for simple LRS. 

As a corollary of the proof of the decidability of Positivity of simple LRS up to order 9 \cite{ouaknine2014positivity}, Ultimate Positivity for simple LRS is \textit{constructively} decidable if one of the following holds: \textbf{(a)} all characteristic roots have the same modulus; {\bf(b)} there are at most three pairs of complex conjugates among the dominant (maximal modulus) characteristic roots. 

We argue that for the original simple LRS $(u_n)_n$, $5$ is the highest order that guarantees that at least one of the conditions holds for the resulting simple LRS $(v_n)_n$ in inequality \ref{eq:criticalcopy}. For this, we recall the property discussed after equation \ref{eq:exppoly}: if $U$ is the set of characteristic roots of $(u_n)_n$, then the set of characteristic roots of $v_n = u_n^2$ is $V = \{\lambda_1\lambda_2: \lambda_1, \lambda_2 \in U\}$. By Proposition \ref{prop:folklore}, $U$ contains a real positive dominant root $\rho$. It is clear that the dominant roots of $V$ result from, and only from multiplying together pairs of dominant roots from $U$. If $U$ does not have complex dominant roots, neither does $V$. If $\lambda \in U$ is a dominant complex root, then $\lambda\bar\lambda = \rho^2$. If $U = \{\rho, \lambda_1, \lambda_2, \bar{\lambda_1}, \bar{\lambda_2}\}$, all dominant, then all roots of $V$ are dominant, and condition \textbf{(a)} is met. The only remaining case is that $U$ has one pair of complex conjugates among its dominant roots: the scenario that results in most dominant roots in V is $U_{dom} = \{\rho, -\rho, \lambda, \bar{\lambda}\}$. Then, the dominant roots in $V$ are $\{\rho^2, -\rho^2, \lambda^2, \bar\lambda^2, \pm \rho\lambda, \pm \rho\bar\lambda\}$: three conjugate pairs, and condition \textbf{(b)} is met. Finally, we record that order 5 is maximal: consider $U = \{\rho, \lambda_1, \lambda_2, \bar{\lambda_1}, \bar{\lambda_2}, \alpha\}$ with $\alpha$ non-dominant. Then $V$ has five pairs of complex conjugates among its dominant roots, along with the presence of non-dominant roots.

\subsection{Non-simple LRS}
We treat order $4$ LRS: our techniques naturally apply to lower orders too. We make extensive use of the real exponential polynomial closed form \ref{eq:realexppoly} and the surrounding discussion. The key lies in expressing the critical inequality \ref{eq:criticalcopy} as
\begin{equation}
\label{eq:start}
\seq{\mathbf{p}, \mathbf{q_n}}^2 - \seq{\mathbf{b_1}, \mathbf{q_n}}^2 - \dots - \seq{\mathbf{b_4}, \mathbf{q_n}}^2 \ge 0 ~~\Leftrightarrow~~ \seq{\mathbf{x}, \mathbf{r_n}} \ge 0
\end{equation} 
and choosing $\{\mathbf{q_n}\}_{n\in\naturals}$ judiciously. If all the characteristic roots of the original LRS are real, then $\mathbf{q_n}$ is free of trigonometric terms, and hence so is $\mathbf{r_n}$. Thus $\seq{\mathbf{x}, \mathbf{r_n}}$ is also an LRS with all real characteristic roots, and constructively deciding the existence of $N_2$ is easily done through elementary growth arguments. We shall thus assume the presence of a pair of complex conjugates among the characteristic roots. As discussed through Proposition \ref{prop:folklore}, any decision regarding Ultimate Positivity is NO in the absence of a real positive dominant root. At order $4$, this means that there is \textbf{exactly one pair of complex conjugates} among the roots. We further assume, without loss of generality, that \textbf{the real positive dominant root is unity}. We shall also assume \textbf{non-degeneracy}, i.e.\ the ratio of any pair of distinct roots of the characteristic polynomial is not a root of unity. This can be detected, courtesy Lemma \ref{lemma:rootofunity}. In our restricted setting, degeneracy can arise because: (a) $-1$ is a characteristic root; (b) a characteristic root is of the form $\rho e^{2\pi i \cdot \frac{\ell}{k}}$, i.e. a scaled $k^{th}$ root of unity. In this case, any LRS $\seq{\mathbf{v}, \mathbf{q_n}}$ with roots $\{1, \alpha, \rho e^{\pm 2\pi i \cdot \frac{\ell}{k}}\}$ can be decomposed as the interleaving of $2k$ real LRS, each with characteristic roots $\{1, \rho^{2k}\} \cup \{\alpha^{2k}\}$.

The only possibility, therefore, is that the characteristic roots are $1, 1, \gamma, \bar{\gamma}$. Let $0 < |\gamma| = \rho \le 1$, where $\gamma = \rho e^{i\theta}$ is not a scaled root of unity. We take inequality \ref{eq:start} as the starting point for our computations. Let $\mathbf{q_n} = \begin{bmatrix} n & 1 & \rho^n\cos(n\theta - \varphi) & \rho^n\sin(n\theta -\varphi) \end{bmatrix}^T$. Let $\mathbf{u_1}^T, \dots, \mathbf{u_4}^T$ be the rows of the \textbf{invertible} matrix $\begin{bmatrix} \mathbf{b_1}& \dots & \mathbf{b_4}\end{bmatrix}$. The table below shows the terms and coefficients on simplifying inequality \ref{eq:start}.

\begin{table}[H]
\begin{tabular}{|l|l|l|}
  \hline
   \textbf{Term}& \textbf{Coefficient}& {\bf Explicitly} \\
  \hline
  $n^2$ & $z_2$ & $p_1^2 - \seq{\mathbf{u_1}, \mathbf{u_1}}$ \\
   \hline
  $n$ & $z_1$ & $2p_1p_2 - 2\seq{\mathbf{u_1}, \mathbf{u_2}}$ \\
   \hline
   $1$ & $z_0$ & $p_2^2 - \seq{\mathbf{u_2}, \mathbf{u_2}} $ \\
  \hline
  $n\rho^n\cos (n\theta - \varphi)$ & $x_2$ & $2p_1p_3 - 2\seq{\mathbf{u_1}, \mathbf{u_3}}$ \\
   \hline
  $n\rho^n\sin (n\theta - \varphi)$ & $y_2$ & $2p_1p_4 - 2\seq{\mathbf{u_1}, \mathbf{u_4}}$ \\
   \hline
   $\rho^n\cos (n\theta-\varphi)$ & $x_1$ & $2p_2p_3 - 2\seq{\mathbf{u_2}, \mathbf{u_3}}$ \\
   \hline
  $\rho^n\sin (n\theta-\varphi)$ & $y_1$ & $2p_2p_4 - 2\seq{\mathbf{u_2}, \mathbf{u_4}}$ \\
   \hline
   $\rho^{2n}$ & $w$ & $\frac{1}{2}(p_3^2 + p_4^2) - \frac{1}{2}
 (\seq{\mathbf{u_3}, \mathbf{u_3}} + \seq{\mathbf{u_4}, \mathbf{u_4}})$ \\
  \hline
  $\rho^{2n}\cos (2n\theta - 2\varphi)$ & $x_0$ & $\frac{1}{2}(p_3^2 - p_4^2) - \frac{1}{2}(\seq{\mathbf{u_3}, \mathbf{u_3}} - \seq{\mathbf{u_4}, \mathbf{u_4}})$ \\
   \hline
  $\rho^{2n}\sin (2n\theta - 2\varphi)$ & $y_0$ & $2p_3p_4 - 2\seq{\mathbf{u_3}, \mathbf{u_4}}$ \\
  \hline
\end{tabular}
\end{table}

If $\rho < 1$, then the dominant growth rate for the problem to be non-trivial is $n^2, n, 1, $ or $\rho^{2n}$. The former cases can be solved with straightforward growth arguments, while the last case results in an order 3 LRS that can easily be dealt with \cite{ouaknine2014positivity,joeljames3}. We thus assume $\rho = 1$. Again, if $z_2 \ne 0$, then decidability is trivial because the dominant growth rate of $n^2$ is dictated by a single term; hence we assume $z_2 = 0$. In this case, there are two groups of terms, based on growth rate: one with $n$, the other with $1$. To study these groups, we define
\begin{align}
f(t) &= z_1 + x_2 \cos(t -\varphi) + y_2\sin(t - \varphi) \\
g(t) &= z_0 + w + x_1\cos(t - \varphi) + y_1\sin(t - \varphi) + x_0 \cos(2t - 2\varphi) + y_0\sin(2t-2\varphi)
\end{align}
Since $\theta$ is not a rational multiple of $2\pi$, $\{n\theta \text{ mod } 2\pi\}$ is dense in $[0, 2\pi]$, and we invoke Lemma \ref{eq:liminfmin} to deduce
\begin{equation}
\liminf_{n\in \naturals} f(n\theta) = \min_{t \in [0, 2\pi]} f(t) = z_1 -\sqrt{x_2^2 + y_2^2} = \mu.
\end{equation}
If $\mu < 0$, then the critical inequality $nf(n\theta) + g(n\theta) \ge 0$ will be violated infinitely often. If $\mu > 0$, we can compute an $N_2$ beyond which it is guaranteed to be satisfied. We thus concern ourselves with the case where $\mu = 0$. Recall the discussion around $\mu$ when its concept was first defined after Proposition \ref{prop:folklore}: it is an intrinsic property of the problem itself, and invariant under the ``phase'' $\varphi$ chosen in the basis of solutions. We thus assume that $\varphi$ is chosen in such a way that the minimum is attained at $\varphi$, i.e. $f(\varphi) = 0$. This choice can be made by applying the trigonometric identity $\cos(a - b) = \cos a \cos b + \sin a \sin b$ to $f(t)$. This means that $y_2 = 0$, and we choose $-z_1 = x_2 < 0$.

Now, if $g(\varphi) > 0$, we compute a positive lower bound on $f(t)$ for $t$ such that $g(t) < 0$. This then results in an $N_2$ beyond which $nf(n\theta) + g(n\theta) \ge 0$ is guaranteed. If $g(\varphi) < 0$, then the inequality has infinitely many violations. This is due to Lemma \ref{eq:quadraticdecay}, which asserts that there are infinitely many $n$ for which $f(n\theta) \le 2\pi^2z_1/5n^2$. These $n$ are necessarily such that $n\theta$ is close to $\varphi$, and the negativity of $g(n\theta)$ is thus decisive.

The final case that remains is $g(\varphi) = 0$. We argue that remarkably, it does not arise at all!
\begin{lemma}
If $z_2 = \mu=0$, it cannot be the case that $f(\varphi) = g(\varphi) = 0$.
\end{lemma}
\begin{proof}
Suppose, for the sake of contradiction, the scenario actually occurs. This means that \\
$
z_2 = z_1 + x_2 = z_0 + w + x_1 + x_0 = 0.
$
From the table, these respectively imply
\begin{align*}
p_1^2 &= \seq{\mathbf{u_1}, \mathbf{u_1}}, \\
p_1(p_2 + p_3) &= \seq{\mathbf{u_1}, \mathbf{u_2} + \mathbf{u_3}}, \\
(p_2 + p_3)^2 &= \seq{\mathbf{u_2} + \mathbf{u_3}, \mathbf{u_2} + \mathbf{u_3}}.
\end{align*}
This implies that $|\seq{\mathbf{u_1}, \mathbf{u_2} + \mathbf{u_3}}| = ||\mathbf{u_1}||\cdot||\mathbf{u_2} + \mathbf{u_3}||$, i.e. $\mathbf{u_1}$ is a scaled multiple of $\mathbf{u_2} + \mathbf{u_3}$. This contradicts the fact that the rows of the invertible $\begin{bmatrix} \mathbf{b_1}& \dots & \mathbf{b_4}\end{bmatrix}$ are linearly independent, and we're done.
\end{proof}

\section{Non-uniform Robustness: Decidability at order four}
\label{section:decidability2}

In this section, we prove Theorem \ref{thm:decide2}. The techniques naturally apply to lower orders, and we omit their explicit treatment. Recall the critical condition from our overview in \S\ref{section:nonuniformoverview}:
\begin{equation}
\mu(\mathbf{c'}) = \liminf_{n\in \naturals}\seq{\mathbf{p'}, \mathbf{q_n}}_{dom} \ge 0
\end{equation}
for all $\mathbf{c'}$ in the neighbourhood is necessary for the decision to be YES; the inequality holding strictly is sufficient. Critical cases arise when the surface of the neighbourhood touches the region where $\mu = 0$, and the non-dominant terms, if any, can potentially have a negative contribution. We demonstrate that these can be detected and dealt with.

Since Proposition \ref{prop:folklore} guarantees the existence of a real positive dominant term, $\seq{\mathbf{p}, \mathbf{q_n}}_{dom}$ can only be of one of the following forms: {\bf(a)} $z$; {\bf(b)} $z + w(-1)^n$; {\bf(c)} $z + x\cos n\theta + y\sin n\theta$; {\bf(d)} $z + x\cos n\theta + y\sin n\theta + w(-1)^n$, where $x, y, z, w$ are linear in the initialisation $\mathbf{c}$. Cases (a), (b), and (c), (d) where $\theta$ is a rational multiple of $2\pi$ (detected with Lemma \ref{lemma:rootofunity}) are the easiest. The region $\mu \ge 0$ is carved out by \textit{finitely} many halfspaces, defined by separating hyperplanes of the form $z + bw + c_0 x + s_0 y = 0$. By elementary linear algebra and co-ordinate geometry (e.g.\. by working in a basis where the neighbourhood is a perfect hypersphere), one can determine whether $\mu > 0$ for the entire neighbourhood, or whether $\mu < 0$ for some points in the neighbourhood, or whether the neighbourhood touches a hyperplane. Each hyperplane has at most one point of tangency, whose algebraic coordinates can be solved for.  These critical points are low-dimensional instances of Ultimate Positivity, and can be decided with the state of the art \cite{ouaknine2014ultimate}.

We therefore assume that $\theta$ is not a rational multiple of $2\pi$, and we are in Case (c) or (d). We apply Lemma \ref{eq:liminfmin}, we get that 
\begin{equation}
\mu(\mathbf{c}) = \liminf_{n \in \naturals} \seq{\mathbf{p}, \mathbf{q_n}} = \min_{t \in \reals, b \in \{\pm 1\}} z + x\cos t + y\sin t + wb = z - \sqrt{x^2 + y^2} - |w|.
\end{equation}
If we are in Case (d), there are no dominant roots, and $\mu \ge 0$ throughout the neighbourhood is necessary as well as sufficient for the decision to be YES. This is an algebraic condition, and can be checked using the First Order Theory of the Reals.\footnotemark

\footnotetext{$|w|$ is expressed with $\exists r.~r^2 = w^2 \land r \ge 0$; a similar trick works for $\sqrt{x^2 + y^2}$.}

Case (c) remains. $\seq{\mathbf{q_n}, \mathbf{p}} = z + x\cos n\theta + y\sin n\theta + w\alpha^n$, where $0 < |\alpha| < 1$. As discussed, we can use the First Order Theory of the Reals to check the sufficient $\mu > 0$, and the necessary $\mu \ge 0$ throughout the neighbourhood. We consider the scenario where the necessity check succeeds, but the sufficiency check fails. The decision can be NO only if there are points on the surface of the neighbourhood where $\mu = 0$, and the non-dominant $w\alpha^n$ can make a negative contribution. We describe how these points are found and analysed. First, we observe that the region $\mu \ge 0$ is given by the cone $z - \sqrt{x^2 + y^2} \ge 0$. It can be intuited as being carved out by a continuum of hyperplanes $z + x\cos\phi + y\sin\phi = 0$. We encode the above discussion to find the critical points with the following first order formula with free variable $c$, which stands for $\cos \phi$
\begin{equation}
\label{eq:intersection}
\chi_1(c):= \exists s \exists \mathbf{c'}.~ (\mathbf{c'} - \mathbf{c})^T\mathbf{S}(\mathbf{c'} - \mathbf{c}) = 1 \land z' + cx' + sy' = 0 \land c^2 + s^2 = 1 \land w' \sim 0.
\end{equation}

In the above $\sim$ is $\ne$ if the non-dominant root $\alpha < 0$, and is $<$ if $\alpha > 0$. We can use Theorem \ref{thm:renegar} to get an equivalent quantifier free formula: this comprises purely of polynomial (in-)equalities in the free variable $c$. The set of $c$, and hence $\cos \phi$, satisfying these, consists of finitely many intervals. Of course, Ultimate Positivity is guaranteed when this set is empty: it means there are no points threatening to violate Ultimate Positivity.

We first dispose of the case where all intervals consist of single points. Consider an interval $\{c_0\}$ consisting of a single point. This is illustrated by the case of the ball touching the cone in Figure \ref{fig:geometricpicture}. Due to its origins and discrete occurrence, $c_0$ must be a root of a polynomial obtained by quantifier elimination on $\chi_1$, and is hence algebraic. The corresponding critical point is the point of tangency of the neighbourhood with a hyperplane with a real algebraic equation. Thus, it generates a real algebraic instance of Ultimate Positivity, which can be decided with the techniques of \cite{ouaknine2014ultimate}.

If, however, the set of $c$ satisfying $\chi_1$ consists of intervals that have more than one point, then the techniques of \cite{ouaknine2014ultimate} to decide Ultimate Positivity for a single point with algebraic coordinates are no longer accessible. This situation is illustrated by the case of the ball nestled in cone in Figure \ref{fig:geometricpicture}. Let $[\phi_1, \phi_2]$ be an interval of $\phi$ such that: a) all values of $c$ between $\cos\phi_1$ and $\cos\phi_2$ satisfy $\chi_1$, b) The corresponding witnesses $z'$ are at most $z_0$, and c) The corresponding witnesses $w'$ have magnitude at least some fixed $w_0$. Then, we must have for each $\phi$ (and corresponding $z(c), x = -cz, y = -z\sin \phi, w)$) in this interval, the following inequality is violated only finitely often:
\begin{equation}
z - z\cos(n\theta - \phi) + w\alpha^n \ge 0.
\end{equation}

We consider an even weaker inequality, which, in this context, we argue is bound to be violated infinitely often:
\begin{equation}
z_0[n\theta - \phi]_{2\pi}^2  \ge 2w_0\alpha^n.
\end{equation}

The argument hinges on Lemma \ref{lemma:existsreal}, which we restate:
\begin{lemma}
\label{lemma:existsreal3}
For every irrational number $x$, strictly decreasing real positive function $\psi$, and interval $\mathcal{I} = [a, b] \subset [0, 1], ~ a \ne b$, there exists $y_0 \in \mathcal{I}$ such that $[nx - y_0] < \psi(n)$ for infinitely many even $n$, and $y_1 \in \mathcal{I}$ such that $[nx - y_1] < \psi(n)$ for infinitely many odd $n$.
\end{lemma}

Now, if $\alpha < 0$, we use Lemma \ref{lemma:existsreal3} on the irrational $\theta/2\pi$, and the decreasing $\sqrt{\frac{w_0 |\alpha|^n}{2\pi^2z_0}}$ to argue that there exists a $\phi$ in the desired interval, such that the weaker inequality will be violated for infinitely many $n$ of the the appropriate parity. Thus, we can return NO if we are in the case where the set of $c$ satisfying $\chi_1$ (equation \ref{eq:intersection}) consists of intervals that contain more than a single point.

\section{Uniform Robustness: Hardness at order five}
\label{section:hardness}
We shall prove Theorem \ref{thm:hardness} in this section. That is, given $t \in \mathcal{T}, s$ as defined in equation \ref{eq:keyset}, we shall give rational $\mathbf{a}, \mathbf{c}$ such that varying $\mathbf{S} =\mathbf{G}^T\mathbf{G}$ while invoking $\mathbf{S}$-Robust Positivity decision procedures will enable us to approximate $L(t, s)$ and $L_\infty(t, s)$ to arbitrary precision.

We assume $t, s$ are specified by $\cos \theta, \cos \varphi \in \realalgebraics$, such that $2\pi t = \theta, 2\pi s = \varphi$. Our LRR $\mathbf{a}$ is such that the roots of the characteristic polynomial are $1, 1, 1, e^{2\pi it}, e^{-2\pi i t}$, i.e. the characteristic polynomial is 
$
(X- 1)^3(X^2 - 2X\cos\theta + 1)
$.

Here, 
$
u_n = \mathbf{e_1}^T\mathbf{A}^n\mathbf{c} = \seq{\mathbf{p}, \mathbf{q_n}}
$. For the problem instance we create in our reduction, we choose $\mathbf{p} = \begin{bmatrix}r & 0 & 1+\frac{r}{2} & -1 & 0 \end{bmatrix}^T$, where $r$ is a parameter we use to tune our guess for $L(t, s)$ and $L_\infty(t, s)$; we choose $\mathbf{q_n}^T = \mathbf{e_1}^T\mathbf{A}^n\mathbf{V} = \begin{bmatrix}n^2 & n & 1 & \cos(2\pi(nt-s)) & \sin(2\pi(nt-s))\end{bmatrix}$. We choose $\mathbf{S} = \mathbf{G}^T\mathbf{G}/r^2$, where $\mathbf{G} = \mathbf{V}^{-1}$. Thus, our critical inequality is
\begin{equation}
\label{eq:corereduction}
\seq{\mathbf{p}, \mathbf{q_n}} = \mathbf{e_1}^T\mathbf{A}^n\mathbf{c} \ge \max_{d \in \mathcal{B}_\mathbf{S}}  \mathbf{e_1}^T\mathbf{A}^n\mathbf{d} = ||r(\mathbf{e_1}^T\mathbf{A}^n\mathbf{G}^{-1})^T|| = r||\mathbf{q_n}||.
\end{equation}

Let $\Psi(n, r)$ denote the proposition $ \seq{\mathbf{p}, \mathbf{q_n}} \ge r||\mathbf{q_n}||$. Our reduction works by proving that for any guess $r>0$, given $\varepsilon>0$, we can compute an $N$ such that for all $n \ge N$
\begin{align}
\label{eq:property1}
&\Psi(n, r) \Rightarrow n[nt - s] > \frac{(1-\varepsilon)\sqrt{7r}}{4\pi}. \\
\label{eq:property2}
\neg &\Psi(n, r) \Rightarrow n[nt - s] < \frac{\sqrt{7r}}{(1-\varepsilon)4\pi}.
\end{align}
To compute $L_\infty(t, s) = \liminf_{n \in \naturals} n[nt-s]$ by increasingly precise approximations, we query Robust Uniform Ultimate Positivity: does $\Psi(n,r)$ hold for all but finitely many $n$? If the decision is YES, then we use property \ref{eq:property1} to argue that for any $\varepsilon$, $n[nt-s]$ exceeds $\frac{(1-\varepsilon)\sqrt{7r}}{4\pi}$ for all but finitely many $n$, hence $L_\infty(t, s)$ must be at least $\frac{\sqrt{7r}}{4\pi}$. Conversely, if the decision is NO, we use property \ref{eq:property2} to deduce that for any $\varepsilon$, $n[nt-s]$ falls short of $\frac{\sqrt{7r}}{(1-\varepsilon)4\pi}$ for infinitely many $n$, hence $L_\infty(t, s)$ must be at most $\frac{\sqrt{7r}}{4\pi}$.

By definition, $L(t,s) = \inf_{n\in\naturals_{>0}}n[nt-s]$. Given the guess $r$, precision $\varepsilon$, the corresponding $N$, and oracle access to whether $\Psi(n, r)$ holds for all $n \ge N$, it follows from properties \ref{eq:property1} and \ref{eq:property2} that we can resolve the dichotomy between $\inf_{n \ge N}n[nt-s] \ge \frac{(1-\varepsilon)\sqrt{7r}}{4\pi}$ and $\inf_{n \ge N}n[nt-s] \le \frac{\sqrt{7r}}{(1-\varepsilon)4\pi}$. By explicitly computing $n[nt-s]$ for the prefix $n < N$ to arbitrary precision, one has a procedure for approximating $L(t, s)$. We now explain how we use Robust Positivity as an oracle to decide whether $\Psi(n, r)$ holds for all $n \ge N$. Note that as it is, our query specifies a recurrence $\mathbf{a}$, an initialisation $\mathbf{c}$, a neighbourhood defined by $\mathbf{S}$ asks for the Robust Positivity of a \textit{suffix} of the sequence, as opposed to the entire sequence. We create a new instance with updated $\mathbf{c'}$ and $\mathbf{S'}$ to implement the shift:
\begin{align}
\forall n\ge N.~ \mathbf{e_1}^T\mathbf{A}^n\mathbf{c} \ge \max_{\mathbf{d} \in \mathcal{B}_\mathbf{S}}\mathbf{e_1}^T\mathbf{A}^n\mathbf{d} ~&\Leftrightarrow~ \forall n.~ \mathbf{e_1}^T\mathbf{A}^n(\mathbf{A}^N\mathbf{c}) \ge \max_{\mathbf{d} \in \mathcal{B}_\mathbf{S}}\mathbf{e_1}^T\mathbf{A}^n(\mathbf{A}^N\mathbf{d}) \\
&\Leftrightarrow~ \forall n.~ \mathbf{e_1}^T\mathbf{A}^n(\mathbf{c'}) \ge \max_{\mathbf{d} \in \mathcal{B}_\mathbf{S'}}\mathbf{e_1}^T\mathbf{A}^n(\mathbf{d'}).
\end{align}
It is clear that $\mathbf{c'} = \mathbf{A}^n\mathbf{c}$. Using the same reasoning as we did in the derivation of equation \ref{eq:bijectivemap}, we argue $\mathbf{S'} = (\mathbf{A}^{-N})^T\mathbf{S}\mathbf{A}^{-N}$. The reduction is thus complete, but for the proof of properties \ref{eq:property1} and \ref{eq:property2}.

By definition, $\Psi(n, r)$ holds if and only if $rn^2 + \frac{r}{2} + 1 - \cos 2(\pi(nt-s)) \ge r\sqrt{n^4 + n^2 + 2} $. Through elementary algebraic manipulations, we can alternately group the terms as
\begin{equation}
\label{eq:pivotal}
1 - \cos (2\pi (nt-s)) \ge \frac{r}{2}\left(\frac{7n^2 + 14}{(n^2 + \sqrt{n^4 + n^2 + 2})(n^2 +4+  \sqrt{n^4 + n^2 + 2})}\right) = r\cdot Q(n)
\end{equation}

We note that in the limit, the ratio of $Q(n)$ to $7/8n^2$ tends to $1$ from below. On the other hand, for small values of $x$, the expression $x^2/2$ is a close over-approximation for $1 -\cos x$. We capture the crucial interdependence in the following technical lemma.

\begin{lemma}
\label{lemma:numerical}
Let $r > 0$. For every $\varepsilon > 0$, we can compute $N$ such that
\begin{enumerate}
\item For all $n\ge N$, $Q(n) > {7(1-\varepsilon)^2}/{8n^2}$.
\item $1 - \cos x < 7r/{8N^2}  ~\Rightarrow~ 1- \cos x \ge (1 - \varepsilon)^2x^2/2$.
\end{enumerate}
\end{lemma}

For some $r, \varepsilon$, let $N$ be computed by Lemma \ref{lemma:numerical}. Consider $n \ge N$. In case $\Psi(n, r)$ holds, property \ref{eq:property1} follows by considering the beginning and end of the chain of inequalities
\begin{equation}
{2\pi^2[nt-s]^2} = \frac{[2\pi(nt-s)]^2_{2\pi}}{2} \ge 1 - \cos(2\pi(nt-s)) \ge r\cdot Q(n) > \frac{7r(1-\varepsilon)^2}{8n^2}.
\end{equation}
Similarly, if $\neg\Psi(n, r)$ holds, we can use Lemma \ref{lemma:numerical} to construct the chain
\begin{equation}
2\pi^2(1-\varepsilon)^2 [nt-s]^2= \frac{(1-\varepsilon)^2 [2\pi(nt-s)]_{2\pi}^2}{2} \le 1 - \cos(2\pi(nt-s)) < r\cdot Q(n) < \frac{7r}{8N^2}.
\end{equation}

\section{Extensions and Perspective}
\label{section:perspective}
We note that our techniques for $\mathbf{S}$-Robust Non-uniform Ultimate Positivity hinge on the First Order Theory of the Reals. Observe that this was rather agnostic to the exact shape of the neighbourhood: we can easily extend the same techniques to arbitrary \textit{semi-algebraic} neighbourhoods. 

As outlined at the outset, we contributed towards a sharp and comprehensive picture of what is \textit{decidable} about Robust Positivity Problems for real algebraic Linear Recurrence Sequences. We find it remarkable that number-theoretic analyses involving Diophantine approximation, which usually show up in the context of hardness, also play a significant role in our \textit{decidability} proofs! However, a rather conspicuous gap in our picture is the status of $\mathbf{S}$-Robust Non-uniform Ultimate Positivity at order $5$: this seems to require even more delicate analysis. 

An obvious, but possibly tedious future direction would be to tie up the book-keeping loose ends, and meticulously account for the complexity of our techniques. We chose to work with algebraic numbers; in settings involving rational numbers where scaling to integers and accessing an $\mathsf{PosSLP}$ oracle is viable, the complexity usually lies in $\mathsf{PSPACE}$. However, this might blow up significantly in the absence of efficient positivity testing for a different class of arithmetic circuit.

At a higher level, we note that we chose our norm to be based on the standard matrix inner product. It is interesting to investigate what kinds of decidability and hardness results hold for neighbourhoods specified using different norms. Perhaps, results could be universal across a wider class of norms, and there could be a profound underlying linear-algebraic reason whose discovery would be mathematically significant.

In the grand Formal Methods scheme, the study of Hyperproperties \cite{hyperproperties} is an exciting natural way robustness problems for Linear Dynamical Systems could fit in. Hyperlogics reason about sets of traces of an infinite time system, rather than a single trace. They gained importance as a means to verify security in view of attacks like Meltdown and Spectre. A quintessential hyperproperty, for instance, would specify a reasonable notion of \textit{indistinguishability} of traces. In that regard, our notions of $\mathbf{S}$-Robust Positivity and $\mathbf{S}$-Robust Uniform Ultimate Positivity bear striking resemblance. Exploring deeper connections is a fascinating future research avenue.

\clearpage
\bibliography{main}
\clearpage

\appendix
\section{Appendix: Notation and Prerequisites}
\label{appendix:prelims}
For the purposes of discussing robustness, we shall use $\ball$ to denote the unit Euclidean ball in $\reals^\kappa$, centred at the origin. Similarly, we use $\ball_{\mathbf{S}}$ to denote the set of $\mathbf{d}$ such that $\mathbf{d}^T\mathbf{Sd} \le 1$. For real column vectors $\mathbf{x}, \mathbf{y}$, we use $\seq{\mathbf{x}, \mathbf{y}}$ to denote the inner product $\mathbf{x}^T\mathbf{y} = \mathbf{y}^T\mathbf{x}$. The notation $||\mathbf{x}||$ denotes the standard $\ell^2$-norm $\sqrt{\seq{\mathbf{x}, \mathbf{x}}}$.

Throughout this paper,  $\naturals$, $\integers$, $\rationals$, $\reals$, and $\complexes$ respectively denote the natural numbers, integers, rationals, reals, and complex numbers. $\alpha \in \complexes$ is said to be algebraic if it is a root of a polynomial with integer coefficients. Algebraic numbers form an algebraically closed field, denoted by $\algebraics$. We denote the field of real algebraic numbers by $\realalgebraics$.

This Appendix contains a brief initiation to this number field $\realalgebraics$ and $\algebraics$. The key takeaways are that the usual arithmetic as well as polynomial root computation can be carried out with perfect precision, and that the First Order Theory of the Reals $\seq{\reals; +, \cdot, \ge, 0, 1}$ is a decidable logical system powerful enough to fit our purposes.

\subsection{Algebraic Numbers: Arithmetic}
For an algebraic number $\alpha$, its defining polynomial $p_\alpha$ is the unique polynomial in $\integers[X]$ of least degree such that the GCD of its coefficients is $1$ and $\alpha$ is one of its roots.
Given a polynomial $p \in \integers[X]$, we denote the length of its representation by $\text{size}(p)$; its height, denoted by $H(p)$, is the maximum absolute value of the coefficients of $p$; $d(p)$ denotes the degree of $p$. The height $H(\alpha)$ and degree $d(\alpha)$ of $\alpha$ are defined to be the height and degree of $p_\alpha$.

For any $p \in \integers[X]$, the distance between distinct roots is effectively lower bounded in terms of its degree and height \cite{mignottecon}.
This bound allows one to represent an algebraic number $\alpha$ as a 4-tuple $(p,a,b,r)$ where $p$ is the defining polynomial, and $a+bi$ is a rational approximation of sufficient precision $r\in\rationals$. We use $\text{size}{\alpha}$ to denote the size of this representation, i.e., number of bits needed to write down this 4-tuple.

Given a polynomial $p\in \integers[X]$, one can compute its roots in polynomial time \cite{findroots1operate1}. Recently, implementations of algorithms to factor polynomials in $\algebraics[X]$ have been verified \cite{factor-algebraic}. Given $\alpha$, $\beta$ two algebraic numbers, one can always compute the representations of $\alpha+\beta$, $\alpha\beta$, $\frac 1 \alpha$, $\Re(\alpha)$, $\Im(\alpha), |\alpha|$, and decide $\alpha = \beta$, $\alpha > \beta$ in polynomial time wrt the size of their representations. \cite{findroots1operate1,findroots2operate2}.

\subsection{First Order Theory of the Reals}
This logical theory reasons about the universe of real numbers, and is denoted $\seq{\reals; +, \cdot, \ge, 0, 1}$. That is, variables take real values; terms can be added and multiplied, we have the comparison predicate, and direct access to the constants $0$ and $1$. Thus, our propositional atoms are inequalities involving polynomials with integer coefficients. With existential quantifiers and polynomials, we can thus express algebraic constants too. Formally, we have access to only the existential quantifier, negation, and disjunction; however, this can express the universal quantifier and all other Boolean connectives as well.

Variables are either quantified or free. Remarkably, the First Order Theory of the Reals admits quantifier elimination: for any formula $\chi(\mathbf{x})$, whose free variables are $\mathbf{x}$, there exists an \textbf{equivalent} formula $\psi(\mathbf{x})$ that does not contain any quantified variables. The following result is relevant to us.
\begin{theorem}[Renegar \cite{renegar}]
\label{thm:renegar}
Let $M \in \naturals$ be fixed. Let $\chi(\mathbf{x})$ be a formula with fewer than $M$ variables in total. There exists a procedure that returns an equivalent quantifier-free formula $\psi(\mathbf{x})$ in disjunctive normal form. This procedure runs in time polynomial in the size of the representation of $\chi$.
\end{theorem} 
\section{Appendix: Ostrowski Numeration System}
\label{appendix:ostrowski}

In this appendix, we prove Lemma \ref{lemma:existsreal}. We state number-theoretic properties of the continued fraction representation and Ostrowski Numeration System without proof. We refer the reader to \cite{bourla2016ostrowski} for a more detailed exposition, and we closely follow the discussion surrounding \cite[Propositions 1.1, 2.1]{berthe2022dynamics} in our own proof. We first prove a slightly simpler statement.

\begin{lemma}
\label{lemma:existsreal2}
For every irrational number $x$, strictly decreasing real positive function $\psi$, and interval $\mathcal{I} = [\alpha, \beta] \subset [0, 1], ~ \alpha \ne \beta$, there exists $y \in \mathcal{I}$ such that $[nx - y] < \psi(n)$ for infinitely many $n$.
\end{lemma}
\begin{proof}
Without loss of generality, we can assume that $x \in (0, 1)$. Consider the continued fraction representation of $x$: $[0; a_1, a_2, a_3, \dots]$
$$
x = \cfrac{1}{a_1 + \cfrac{1}{a_2 + \cfrac{1}{a_3 + \cfrac{1}{\ddots}}}}
$$
where $a_1, a_2, a_3, \dots \in \naturals$. Let the rational approximation of $x$ obtained by truncating the expansion at the $k^{th}$ level be $\frac{p_k}{q_k}$, i.e. $\frac{p_1}{q_1} = \frac{1}{a_1}$, and so on. Let $\theta_k = q_k x -p_k$. We have that $|\theta_k| = (-1)^k\theta_k$. It is well known that $|\theta_k| < 1/q_k$. We define $q_{-1} = p_0 := 0$, and $p_{-1} = q_0 := 1$, so that for $k \ge 1$, the following recurrences hold:
$$
p_k = a_kp_{k-1} + p_{k-2}, ~ q_k = a_kq_{k-1} + q_{k-2}
$$
We thus have that $q_k \ge \left(\frac{1 + \sqrt{5}}{2}\right)^k = \phi^k$.

\begin{proposition}[\cite{berthe2022dynamics}]
\label{prop:absconv}
Let irrational $x$ and its continued fraction representation $[0; a_1, a_2, a_3, \dots]$ be as above. The infinite series 
$$
\sum_{i=1}^\infty a_i |\theta_{i-1}|
$$
converges.
\end{proposition}

\begin{proposition}[Ostrowski Numeration System, \cite{berthe2022dynamics}]
\label{prop:numsys}
Every real number $y \in [0, 1)$ can be written uniquely in the form
$$
y = \sum_{i=1}^\infty b_i |\theta_{i-1}| = \sum_{i=1}^\infty (-1)^{i-1}b_i \theta_{i-1} 
$$
where $b_i \in \naturals$ $b_i \le a_i$ for all $i \ge 1$. If for some $i$, $a_i = b_i$, then $b_{i+1} = 0$. $a_i \ne b_i$ for infinitely many odd, and infinitely many even indices $i$.
\end{proposition}

We prove Lemma \ref{lemma:existsreal2} by using the free choice of $b_i$ in this system to construct appropriate $y$. We first handle the issue of placing $y$ in the correct interval $[\alpha, \beta]$. Let $\beta - \alpha = \delta$. We use Proposition \ref{prop:absconv} to argue that there exists a suffix of the infinite series, such that changing the suffix does not change the real number it represents by more than $\delta/2$. Then, we can simply fix the corresponding prefix of $(\alpha + \beta)/2$ to be the prefix of $y$.

Once this prefix is locked in, our strategy is to set $b_i$ to $0$ in even positions, and $1$ in some odd positions, to ensure that for sufficiently large $k$, $n_k = \sum_{i=1}^k b_i (-1)^{i-1}q_{i-1}$ is positive, and increasing in $k$.

Now, notice that since $b_i, p_i$ are all integers, for any $y$,
\begin{align*}
[n_kx - y] &= \left[\sum_{i=1}^k b_i (-1)^{i-1}q_{i-1}x - \sum_{i=1}^k b_i (-1)^{i-1}p_{i-1} - y\right] \\
&= \left[\sum_{i=1}^k b_i (-1)^{i-1}\theta_{i-1} - y\right] \\
&= \left[- \sum_{i=k+1}^\infty b_i (-1)^{i-1}\theta_{i-1} \right] = \sum_{i=k+1}^\infty b_i |\theta_{i-1} |\\
&< \sum_{i=k+1}^\infty b_i \frac{1}{q_{i-1}} \le \sum_{i=k+1}^\infty b_i \frac{1}{\phi^{i-1}} \le \frac{c}{\phi^k}
\end{align*}

Note that the last constant $c$ can be set independently of the choice of which $b_i$ are $1$, and which are $0$: it comes from the convergence of the geometric sum. We now make the choice of where to set $b_i=1$. To conclude the proof, we shall show that given a decreasing function $\psi$, we can ensure that for infinitely many distinct $n_k$, 
$$[n_k x - y] < \frac{c}{\phi^k} \le \psi(n_k) = \psi\left(\sum_{i=1}^k b_i (-1)^{i-1}q_{i-1}\right)$$.

The first inequality is guaranteed. Suppose the second inequality does not hold. Then, from $i = k$ onwards, we keep assigning $b_i := 0$. This holds $n_k$ constant as $k$ increases, but decreases $\frac{c}{\phi^k}$. Eventually, the second inequality will indeed hold. After this point, for the next odd $i$, we can set $b_i$ to $1$, and get a new $n_k$. We continue this ad infinitum, and we are done.
\end{proof}

Now, to get infinitely many \textit{even} $n$, apply Lemma \ref{lemma:existsreal2} with $2x$, $[a, b]$, $\psi_0(n) = \psi(2n)$. For some choice of $y$, there will be infinitely many $n$ such that $[2nx - y] < \psi_0(n) = \psi(2n)$. To get infinitely many \textit{odd} $n$, we can take a subset of the interval, and shift it by $x$. Take $\psi_1(n) = \psi(2n+1)$. For some choice of $y - x$, there will be infinitely many $n$ such that $[n(2x) - (y-x)] = [(2n+1)x - y] < \psi_1(n) = \psi(2n+1)$.

\end{document}